\documentclass{article}
\usepackage{graphicx} % Required for inserting images
\usepackage{tikz}
\usepackage{hyperref}
\usetikzlibrary{patterns} 
\usepackage{amssymb}

\usepackage{amsmath}
\usepackage{subcaption}
\usepackage{booktabs}
\usepackage{caption}
\usepackage{todonotes}
\usepackage{multirow}
\usepackage{amsfonts}
\usepackage{rotating}
\usepackage{url}
\usepackage{color}
\usepackage{xcolor}
\usepackage{array}
\usepackage{mathtools}
\usepackage{comment}
\usepackage{dblfloatfix}
\usepackage{enumitem}
\usepackage{hyperref}

\usepackage{algorithm}
\usepackage{algpseudocode}  % use this instead of algorithmic
\usepackage[utf8]{inputenc}

\usepackage[margin=1.1in]{geometry}
\usepackage{multirow}
\usepackage{amsthm}
\usepackage{amsmath}

\usepackage[official]{eurosym}%€
\usepackage{pgfplots}
\usepackage{tablefootnote}

\usepackage[affil-it]{authblk}  % Paket für Affiliations
\usepackage[numbers,square,sort&compress]{natbib}

\newtheorem{defi}{Definition}
\graphicspath{{Graphics/}}
\newcolumntype{L}[1]{>{\raggedright\arraybackslash}p{#1}}

\newtheorem{theo}[defi]{Theorem}
\newtheorem{lem}[defi]{Lemma}

\newcommand{\NP}{$\mathsf{NP}$-hard}
\newcommand{\Pclass}{$\mathsf{P}$}
\usetikzlibrary{shapes.geometric}
\title{Why Districting Becomes \NP}
\author[1]{Niklas Jost\thanks{Corresponding author: \href{mailto:niklas.jost@tu-dortmund.de}{niklas.jost@tu-dortmund.de}}}
\author[2]{Adolfo R. Escobedo}
\author[1,3]{Alice Kirchheim}

\affil[1]{Chair of Material Handling and Warehousing, TU Dortmund University, 44227 Dortmund, Germany}
\affil[2]{Edward P. Fitts Department of Industrial and Systems Engineering, North Carolina State University, Raleigh, North Carolina 27695, United States of America}
\affil[3]{Fraunhofer Institute for Material Flow and Logistics, 44227 Dortmund, Germany}

\date{\today}

\begin{document}
\maketitle
\begin{abstract}
This paper investigates \emph{why} and \emph{when} the edge-based districting problem becomes computationally intractable. The overall problem is represented as an exact mathematical programming formulation consisting of an objective function and several constraint groups, each enforcing a well-known districting criterion such as balance, contiguity, or compactness. While districting is known to be \NP\ in general, we study what happens when specific constraint groups are relaxed or removed. The results identify precise  boundaries between tractable subproblems (in \( \mathcal{P} \)) and intractable ones (\NP). The paper also discusses implications on node-based analogs of the featured districting problems, and it considers alternative notions of certain criteria in its  analysis.   
\end{abstract}

\vspace{1em}

\section{Introduction}

\textit{Districting} defines a general class of problems for partitioning a set of \textit{basic units} (e.g., zip codes, customers) comprising a geographical region into small clusters according to definite planning criteria. In addition to imposing the structural requirements of such  partitions (i.e., exclusive and exhaustive allocation of the basic units among a fixed set of districts), districting problems tend to emphasize the enforcement of three basic criteria: \textit{balance} -- service demand should be allocated evenly among the districts; \textit{compactness} -- districts should be nearly square or round-shaped, non-distorted, without holes, and with smooth boundaries \citep{butsch2014districting}; and \textit{contiguity} -- it should be possible to move between any two basic units of the district without having to leave the district. 

The formal requirements of a districting plan are encapsulated via deterministic optimization models, for which customized solution approaches are usually developed in tandem. Owing to the inherent computational intractability of districting problems, the predominant emphasis has been on heuristics that can handle large-scale instances (e.g., \cite{bozkaya2003tabu}, \cite{mehrotra1998optimization} \cite{rios2021location}) but lack formal optimality guarantees. Exact algorithmic approaches have also been developed, but they are notoriously difficult to scale. It is well known that the overall problem is \NP. However, the precise role that individual requirements play in contributing to this complexity has not been systematically investigated. In particular, it remains unclear which specific combinations of districting design criteria render the problem computationally intractable.

Understanding which combinations of requirements  cause computational intractability has important practical consequences, as decision-makers must weigh the relevant trade-offs: either insist on enforcing all criteria and face costly or potentially infeasible computations, settle for quick solution approaches (i.e., heuristics) to the detriment of formal optimality guarantees, or relax certain constraints to allow efficient computation at the expense of model fidelity. A clear classification of which requirements render the problem tractable is therefore essential for informed, principled decision-making. This paper aims to fill this gap by listing common design criteria and analyzing the computational complexity of problems that result from enforcing different combinations of these requirements. Our focus is on problems where the edges are the basic units to be partitioned, motivated by logistics applications where strategic territory design is needed to facilitate to travel within the districts in subsequent operational activities (e.g., maintenance, patrolling, etc.). We provide a complete classification for this problem class of  interest, proving for each subproblem whether it is solvable in polynomial time or whether it is \NP.

%% SAVE 
%% Our focus is on districting problems that are conveniently defined on graphs and characterized by discrete compactness measures, which are germane to logistics applications where it is necessary to travel within the districts in subsequent operational activities (e.g., maintenance, patrolling, etc.); this contrasts with continuous and shape-based measures, which are prominently featured in political contexts  \citep{kalcsics2019districting,ric13pol}. We provide a complete classification for the class of problems of interest, proving for each variant whether it is solvable in polynomial time or whether it is \NP. 

The remainder of this paper is structured as follows.  Section~\ref{sec:Re} reviews related work, notably on  districting, node-based/edge-based variants of this class of problems, and  differing related notions of contiguity. Section~\ref{sec:Pre} introduces the necessary preliminaries and formal definitions. Section~\ref{sec:Result} presents the core results, which are supported by a detailed proof framework. The proofs are organized into three parts: (a) variants that are solvable in polynomial time, which are presented in Section~\ref{sec:P}; (b) variants that are NP-complete, which are detailed in Section~\ref{sec:np}; and (c) reductions showing how certain variants generalize or subsume others, which are featured in Section~\ref{sec:Reduc}. Section~\ref{sec:Bonus} considers the minimal imbalance necessary when contiguity is jointly imposed in the two-district case such that there is always a valid solution; a strict bound is provided for this special case. Finally, Section~\ref{sec:con} concludes the paper with a summary of the findings and directions for future research.

\section{Related Work}\label{sec:Re}

A expansive overview of districting models and algorithmic approaches developed to solve them is provided by Kalcsics \& Ríos-Mercado~\cite{kalcsics2019districting}, who present the underlying problems as a generalization of classical location problems, enriched with various general and domain-specific considerations. These models typically aim to allocate voters or geographic units to a limited number of centers while satisfying additional design criteria or requirements (i.e., constraints). Districting problems have been associated historically with political contexts, where they are utilized to define the boundaries of single-member voting precincts. In this setting, satisfying each of the three basic planning criteria described in the preceding section is either mandated by law or regarded as being indispensable for ensuring fair and representative outcomes in the division of a state's population units (e.g., subdivisions or census tracts) \citep{hoj96opt}. Objective criteria for district selection is essential, since systematic manipulation of districts by a party (i.e., gerrymandering) can strongly influence election outcomes \citep{apollonio2009bicolored}. However, relaxing some of the districting requirements is often a necessary step towards obtaining a workable (though not necessarily feasible) districting design \citep{ric13pol}. Such considerations give way to questions regarding computational complexity, which are also germane when considering the opposing  perspective. For example, determining an optimal gerrymandered solution is $\mathcal{NP}$-hard, even for relatively simple graph structures such as trees with diameter 4 \citep{bentert2023complexity}.  

In recent years, districting methodologies have also gained importance withing logistics contexts, where they are implemented towards the strategic design of regions where a specific type of service is to be subsequently provided \citep{mou17upd}. Logistics districting applications are wide-ranging and include police patrolling, delivery of emergency services, and facilities planning \citep{kal05tow}. In these settings, the three design criteria are said to induce practical benefits such as shorter travel times, but their enforcement is less stringent relative to political contexts. For example, deviations of $10\%$ to $20\%$ from perfectly balanced districts are commonly considered (e.g., see \cite{daskin1997network} and \cite{noh2025inbound}). Moreover, logistics districting can be interpreted as a generalization of various \textit{location-allocation problems}, which seek to find the optimal placement of a set of facilities to serve a given set of demand (e.g., the $p$-median \citep{ave07com} and extensions thereof \citep{she00dis}). In other words, removing certain basic districting requirements can yield other classic optimization problems that are of wider interest. Thus, the classification undertaken in this work could bring more direct benefits to these contexts.

The focus of this work is on districting problems for logistics contexts, which are typically defined on planar graphs. While certain results could be applicable to political districting, the two broad categories of problems are not fully interchangeable. One significant distinction is the choice of basic units to be partitioned. Whereas graph-based political districting models have focused exclusively on nodes as the basic units to be partitioned \citep{ric13pol}, logistics districting models can also feature the partitioning of edges, directed (e.g., see \cite{butsch2014districting}) or undirected (e.g., see \cite{kassem2023models}). Another key distinction is the predominance of continuous geometric distances  in political districting, as opposed to discrete network-based distances in logistics districting, which more  accurately capture geographical considerations \citep{car12doe,zol83sal}. 

The distinguishing characteristics of the featured logistics districting problem shape the complexity analyses performed in this work. Hence, these defining aspects are covered in more detail in the remainder of this section; others are elaborated in the subsequent section.  

\subsection{A Brief Comparison of Node-based and Edge-based Districting Problems}

Most existing logistics districting problems focus on the nodes or vertices of the graph as the basic units to be partitioned. Their goal is to designate a fixed number of vertices as the \emph{centers} of the districts and to allocate all other vertices (along with their demands for a certain service) to the selected centers.  Node-based problems  generalize certain facility location-location problems, such as the $p$-median or $p$-center, extending them by incorporating districting-specific criteria (e.g., balance, contiguity). The reader is directed to Kalcsics \& Ríos-Mercado~\cite{kalcsics2019districting} for a survey of various models and algorithms that have been introduced for node-based districting across various logistics applications. 

The computational complexity of node-based districting has been examined in various settings. Dharangutte et al.~(\cite{dharangutte2025hardness} consider a model where districts must balance multiple community types. They prove that many variants remain $\mathsf{NP}$-hard even on restricted graph structures such as trees or planar graphs, while also presenting approximation algorithms for specific cases. Cohen-Addad et al.~\cite{cohen2020computational} explore the computational hardness of node-based districting more broadly including establishinging strong inapproximability results. The authors also identify tractable instances, specifically providing a fully polynomial-time approximation scheme (FPTAS) for complete graphs and trees, and they design greedy algorithms with provable guarantees under constraints such as bounded degrees, binary vertex weights, or limited district sizes.

In contrast to node-based models, edge-based districting treats edges---representing road segments, communication links,  flow paths, etc.---as the basic units  to be allocated to districts. Note, however, that the centers are still designated from the graph's vertices. The edge-based approach more precisely captures the focus of various practical scenarios, such as postal delivery and maintenance activities, where the road network infrastructure forms the basis for partitioning \citep{assad95arc,kan97des}. Algorithmic approaches for these types of logistics districting problems remain relatively underexplored in comparison to their node-based counterparts. Here, the focus has been on heuristic  approaches, which are well-suited for large-scale
problem instances but lack formal optimality guarantees (e.g., see \cite{butsch2014districting}). Exact approaches have focused on cut-generation schemes for implementing a prevalent class of contiguity constraints with exponential cardinality (e.g., see  \cite{garcia2016novel}).  

Studies relating to the computational complexity of edge-based districting problems are also relatively scant. One notable exception is Bang-Jensen et al.~\cite{ban22com}, who consider several digraph problems with a common task of partitioning arcs in two distinct sets $A_1,A_2$ such that the digraphs $D_1=(V,A_1)$ and $D_2=(V,A_2)$ each satisfy a specific property. The pairs of properties considered are generated from a set of 15 properties of interest, some of which relate to districting design criteria. This includes connectedness and the requirement that each partition has at least or at most a certain number of arcs, which is a simplified version of the balance criterion in districting. However, the analysis in \cite{ban22com} considers imposing only one property at a time for each set, and thus it is significantly different from districting, which requires multiple design criteria and structural requirements to be simultaneously satisfied by $p\ge2$ partitions. Another relevant study is  Cordone \& Maffioli~\cite{cordone2004complexity}, which investigates the complexity of graph tree partitioning problems under various objectives. The authors show that certain objectives, such as min–max and max–min, or constraints like \textit{Inclusion} and \textit{Weight}, render the problems particularly challenging from a computational complexity perspective.  
In contrast to our work, these results focus on minimizing spanning trees within each partition, rather than minimizing the sum of distances to a designated center node.

%The present work contributes mainly to the understanding of the edge-based logistic districting problem. It also provides supplementary insights into node-based districting by exploiting transformations that, in certain cases, connect the two problem types. However, one key  difference between the two variants lies in the notions of contiguity, which are described in the next subsection. This distinction makes the featured results only partially transferrable to node-based problems. These specific cases are highlighted throughout the paper. 

The presented work contributes to the understanding of the edge-based districting problem. Nevertheless, all results can be transferred to the node-based districting setting, depending on the specific definition of contiguity adopted in that context; the next paragraphs elaborate on the definition utilized herein. More specifically, most of the provided proofs can be directly extended to the node-based case. In cases where this transfer is not immediate, additional reasoning is provided to demonstrate how the result carries over following the proof for the edge-based case. 

%In short, the most common notion of node-based contiguity is overly restrictively in special cases which does not contribute to real-world solutions. 

%This distinction makes the results partially transferrable to node-based districting. 

%In certain cases, node-based problems can be transformed into edge-based problems (and vice versa), allowing for transferability of our results. Some of these cases will be discussed in the presented analyses to elucidate these connections.

% In order to address these more significant gaps in the literature, this work focuses on edge-based districting\textcolor{red}{Alternative Proposal: edge-based districting as the more general framework. The difference between these variants only lies in the definition of contiguity as described in the next subsection. The most common node-based districting variant of defining contiguity is restrictively in special cases which does not contribute to real-world solutions.} To lend further support for this choice, it is important to underscore that a logistics districting problem instance should prove more computationally demanding when the partitioning is based on the edges rather than on the nodes, according to the differences in cardinality of the respective sets in most real-world instances. In addition, in certain cases, node-based problems can be transformed into edge-based problems (and vice versa), allowing for transferability of our results. Some of these cases will be discussed in the presented analyses to elucidate these connections.

\subsection{Notions of Districting Contiguity}

Balance and compactness are typically defined analogously in node-based and  edge-based models. However, modeling \emph{contiguity} in node-based settings has led to several different approaches in the literature:

\begin{enumerate}
    \item \textbf{Geometric Contiguity}: districts are contiguous if they share a boundary in a non-graph-based spatial representation (see~\cite{kalcsics2019districting}).
    \item \textbf{Vertex-Induced Subgraph Connectivity}: a district consisting of a vertex set $T \subseteq V$ must induce a connected subgraph $G[T]$ (see~\cite{dharangutte2025hardness}).
    \item \textbf{Edge-Based Connectivity for Assignments}: there exists an assignment of edges to districts such that the edges corresponding to each district form a connected subgraph connecting all the assigned vertices (see~\cite{cordero2025optimizing}).
\end{enumerate}

Definitions 2 and 3 are directly applicable in graph-theoretic settings, and any solution satisfying \textit{vertex-induced subgraph connectivity} also satisfies \textit{edge-based connectivity for assignments} but not vice versa. To understand the difference, consider a star graph as in Figure \ref{node-dist}: under Definition 2, a valid solution may not exist since the center node must appear in all districts to ensure connectivity. In contrast, under Definition 3, one can easily construct connected districts using edge assignments. 

\iffalse
\begin{figure}
\centering
\begin{tikzpicture}[every node/.style={circle,draw,minimum size=8mm,inner sep=0pt}]
    % --- First star graph ---
    % center
    \node[fill=green!70] (c1) at (0,0) {};
    % leaves
    \node[fill=purple!70] (a1) at (0,2) {};
    \node[fill=green!70]  (b1) at (2,1) {};
    \node[fill=purple!70] (d1) at (-2,2) {};
    \node[fill=green!70]  (e1) at (-1,-1) {};
    \node[fill=purple!70] (f1) at (2,-2) {};
    
    % edges
    \draw (c1)--(a1);
    \draw (c1)--(b1); % highlight one edge blue
    \draw (c1)--(d1);
    \draw (c1)--(e1);
    \draw (c1)--(f1);

    % --- Second star graph (shifted right) ---
    % center
    \node[fill=green!70] (c2) at (6,0) {};
    % leaves
    \node[fill=purple!70] (a2) at (6,2) {};
    \node[fill=green!70]  (b2) at (8,1) {};
    \node[fill=purple!70] (d2) at (4,2) {};
    \node[fill=green!70]  (e2) at (5,-1) {};
    \node[fill=purple!70] (f2) at (8,-2) {};
    % edges
    \draw[purple] (c2)--(a2);
    \draw[green] (c2)--(b2); % highlight one edge blue
    \draw[purple] (c2)--(d2);
    \draw[green] (c2)--(e2);
    \draw[purple] (c2)--(f2);

\end{tikzpicture}
\caption{The left graph violates \textit{vertex-induced subgraph connectivity} as the purple vertices are not connected; the right graph satisfies \textit{edge-based connectivity for assignments} as there is an allocation of edges such that all purple (and green) vertices are connected to each other.}\label{node-dist}
\end{figure}
\fi
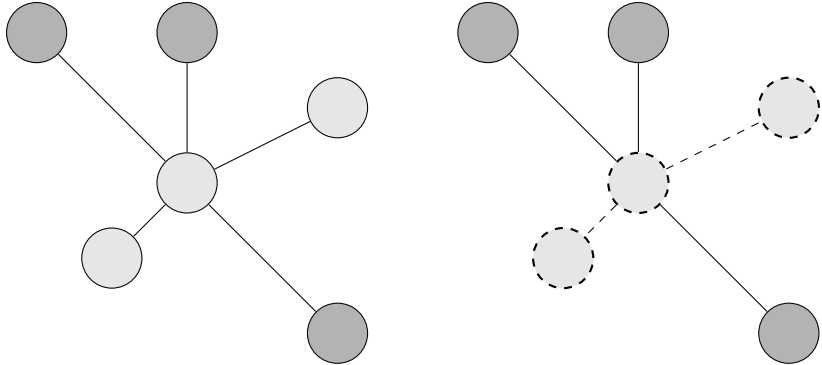
\begin{figure}
\centering
\begin{tikzpicture}[every node/.style={circle,draw,minimum size=8mm,inner sep=0pt}]
    % --- First star graph ---
    % center
    \node[fill=gray!20] (c1) at (0,0) {};
    % leaves
    \node[fill=gray!60] (a1) at (0,2) {};
    \node[fill=gray!20]  (b1) at (2,1) {};
    \node[fill=gray!60] (d1) at (-2,2) {};
    \node[fill=gray!20]  (e1) at (-1,-1) {};
    \node[fill=gray!60] (f1) at (2,-2) {};
    
    % edges
    \draw (c1)--(a1);
    \draw (c1)--(b1); % mark special edge
    \draw (c1)--(d1);
    \draw (c1)--(e1);
    \draw (c1)--(f1);

    % --- Second star graph (shifted right) ---
    % center
    \node[fill=gray!20, dashed, thick] (c2) at (6,0) {};
    % leaves
    \node[fill=gray!60] (a2) at (6,2) {};
    \node[fill=gray!20, dashed, thick]  (b2) at (8,1) {};
    \node[fill=gray!60] (d2) at (4,2) {};
    \node[fill=gray!20, dashed, thick]  (e2) at (5,-1) {};
    \node[fill=gray!60] (f2) at (8,-2) {};
    
    % edges
    \draw (c2)--(a2);
    \draw[dashed] (c2)--(b2); % special allocation
    \draw (c2)--(d2);
    \draw[dashed] (c2)--(e2);
    \draw (c2)--(f2);

\end{tikzpicture}
\caption{The left graph violates \textit{vertex-induced subgraph connectivity} as the darker vertices are not connected; the right graph satisfies \textit{edge-based connectivity for assignments} as there is an allocation of edges such that all vertices of the same type are connected to each other.}
\label{node-dist}
\end{figure}

These distinctions are important not only theoretically but also practically:

\begin{itemize}
    \item \emph{Variant 2} is suitable for political redistricting, where each district must be a geographically or administratively connected unit.
    \item \emph{Variant 3} is more appropriate in logistics contexts, where connectivity is often motivated by the need of continuous travel within each district.
\end{itemize}

Based on the scope of this work, we adopt Variant 3. This choice helps avoid certain pathological cases while retaining theoretical richness and practical relevance. The above star graph example helps illustrate how Variant 2 can create artificial infeasibility in simple structures, potentially complicating analysis without offering practical benefits.

\section{Preliminaries}\label{sec:Pre}

In this section, we formalize the full edge-based districting problem (EBD, for short) considered in this work; the implications on a corresponding node-based districting problem are discussed following each respective proof. The goal of EBD is to select a set of center nodes and partition the edges of a given graph such that each edge is allocated to exactly one center. This assignment defines a set of districts, each centered at a designated node.

We begin by introducing the necessary parameters and decision variables, giving way to a formal mixed-integer programming (MIP) formulation of the problem. Then, to investigate the computational complexity of edge-based districting, we analyze simpler subproblems. By systematically omitting specific constraints, we aim to delineate clear boundaries between polynomial-time solvable cases and \NP\ variants. For this purpose, we group constraints according to the properties they enforce (including balance, contiguity, and compactness), and we define a corresponding notation for each subproblem variant. This approach enables a structured discussion of which constraints contribute most to computational hardness.

\subsection{Parameters, Variables, and Other Formulation Components}

To formally define the edge-based districting problem, it is first necessary to specify the key assumptions and elements of the formulation, including sets, parameters, and decision variables. These elements are summarized in Table~\ref{tab:params}.

The problem is defined on a graph \( G = (V, E) \), where \( V \) denotes the set of nodes (e.g., intersections of roads) and \( E \subseteq V \times V \) represents the set of undirected edges between nodes. Each edge \((j,k) \in E\) may represent a connection such as a road and has a parameter \( b_{(j,k)} \) reflecting length, traffic volume, or other relevant quantity. The task is to partition the edge set \( E \) into \( p \) districts, each represented with a center node. The partitioning is governed by the binary allocation variable \( x_{i,(j,k)} \), which indicates whether edge \((j,k)\) belongs to district \( i \), and the location variable \( w_{ik} \), which indicates whether node \( i \) is selected as a center node for district $k$. As a special case, we consider the \textit{unweighted} case, in which $b_{(j,k)}=1$ for any \((j,k) \in E\). 

The contiguity of each district is enforced by guaranteeing that no subset of edges $D\subseteq E$ allocated to the same district forms an isolated component disconnected from the center node. Specifically, if a subset of edges within a district is not incident to the center node ($c_i\not \in V_D$), there must exist at least one edge in the district that connects this subset to the remainder of the graph. This ensures that all edges in a district are part of a single connected component linked directly or indirectly to its center node.

The districts must also be balanced, which requires that each district contains approximately the same total edge weight. We set a lower bound on the necessary edge weight per district by $\Phi_l$ and an upper bound by $\Phi_u$. In a subsequent section, we discuss variants of this requirement, to  allow two types of deviation from the average number of edges per district, given by  \(\overline{b}\), based on a tolerance parameter \( \tau \).

Finally, the dispersion within the graph is encoded via a distance (i.e., cost)  function \( d \). The basic distance between two vertices, denoted \( d_{i,j} \), is defined as the shortest path length between vertices \( i \) and \( j \). We extend this notion to define the distance between a vertex and an edge, to quantify the distance from a center node \( i \) to an edge \( (j,k) \), which is central for measuring (and minimizing) the dispersion of the districts.

There are multiple ways to define \( d_{i,(j,k)} \). Throughout this paper, we assume that \( d \) reflects shortest-path distances and therefore satisfies the properties of a metric. Intuitively, we interpret \( d_{i,(j,k)} \) as the shortest distance from node \( i \) to edge \( (j,k) \), possibly via an intermediate edge \( (j',k') \) on a shortest path. Formally, 
\begin{equation}
  d_{i,(j,k)} = d_{i,j'} + b_{(j',k')} + d_{k',(j,k)} 
\quad \Leftrightarrow \quad 
\text{$(j',k')$ lies on a shortest path from $i$ to $(j,k)$}.\label{firsteq}  
\end{equation}

Without loss of generality, in equation \eqref{firsteq} it is assumed that $j'$ is closer to $i$ than $k'$; otherwise the indices are flipped. The symmetry of the distance is also required: for any edge \( (i,j) \), it should hold that \( d_{i,(i,j)} = d_{j,(i,j)} \). A general definition of \( d_{i,(j,k)} \) that satisfies these required properties is:
\[
d_{i,(j,k)} := \min(d_{i,j}, d_{i,k}) + \alpha\, b_{j,k},
\]
for any \( \alpha \in [0,1] \). A common choice is \( \alpha = 0 \), as used by Kassem \& Escobedo~\cite{kassem2023models}, which corresponds to the smaller of two node-to-node distance values (i.e., between the center node and each endpoint of $(j,k)$). Alternatively, setting \( \alpha = 0.5 \) might be more intuitive, to reflect that the service location is, on average, at the midpoint of the edge \citep{kal15dis}.

\begin{table}[h]
\centering
\caption{Model Parameters and Decision Variables}
\begin{tabular}{|l|l|p{9cm}|}
\hline
\textbf{Symbol} & \textbf{Type} & \textbf{Description} \\
\hline
$V$ & Set & Set of nodes (e.g., districts) \\
$E \subseteq V \times V$ & Set & Set of edges (undirected connections between nodes) \\
$V_D $ & Set & Set of vertices incident to at least one edge in \( D\subseteq E \), that is, \(V_D:=\{ v \in V \mid \exists (v,u) \in D \text{ or } (u,v) \in D \} \) \\
$\delta(j)$ & Set & Set of edges incident to node $j\in V$ \\
$p\leq |E|$ & Parameter & Total number of districts \\
%$\overline{b}$ & Parameter & Average (or target) total edge weight per district \\
%$\tau $ & Parameter & Tolerance for district size imbalance \\
$\Phi_l$ & Parameter & Minimal number/weight of edges required per district \\
$\Phi_u$ & Parameter & Maximal number/weight of edges required per district \\
$b_{(j,k)}$ & Parameter & Weight of edge $(j,k)$ \\
$d_{v,(j,k)}$ & Parameter & Distance (i.e., cost) between edge $(j,k)$ and vertex $v\in V$ \\

\hline
$x_{i,(j,k)} \in \{0,1\}$ & Decision Variable & 1 if edge $(j,k)$ is allocated to district $i$; 0 otherwise \\
$w_{ik} \in \{0,1\}$ & Decision Variable & 1 if node $k$ is center node of district $i$; 0 otherwise \\
\hline
\end{tabular}
\label{tab:params}
\end{table}

\textbf{Exact Mathematical Formulation}

We now present a (quadratic) binary programming formulation for edge-based districting (BP-EBD). The model incorporates the three before mentioned basic properties: balance, contiguity, and compactness. BP-EBD differs slightly from the original EBD model introduced by Kassem \& Escobedo~\citep{kassem2023models}, based on certain key differing assumptions elaborated later in this section. This includes the presented objective function expression, which can nonetheless be linearized by introducing additional variable indices and constraints (see the proof of Theorem \ref{LPTheorem}). However, that formulation is more cumbersome and, thus, the presented version is maintained in favor of simplicity. 

To facilitate a modular understanding of the problem's complexity, we group the constraints according to the property they enforce. Multiple constraints within each group are then indexed hierarchically. For example, $S_{1}$ and $S_{2}$ are \textit{structural requirements}. This format supports the understanding and notation of the subcases and allows us to define and analyze meaningful subproblems by selectively omitting constraints corresponding to specific design criteria/requirements. Each group of constraints is labeled with a prefix indicating the associated property:

\begin{itemize}
    \item[(S)] Structural requirements are necessary constraints for any meaningful problem.
    \item[(B)] Balance requires that each district should be allocated an approximately equal share of the total edge weight, where the allowable deviations are specified via a tolerance parameter \( \tau \).
    \item[(C)] Contiguity requires that each district should form a connected subgraph, thereby preventing allocated units to be isolated from others in the same partition (and from their center).
    \item[(I)] Integrity enforces that each edge should be allocated to exactly one district.
    \item[(N)] Node center selection indicates whether the centers are pre-selected or chosen as part of the solution process. 
    \item[(O)] Satisfying the Objective of compactness requires that the cumulative distance from the edges to their assigned centers should be as small as possible, so as to promote geographically tight districts.
    \item[(W)] Weightiness enables each edge to have arbitrary positive weights, generalizing the unweighted case where all edges have unit weight.
\end{itemize}

The complete formulation is given by: 

\[ \min \sum_{v\in V}\sum_{i=1}^p\sum_{(j,k)\in E} w_{iv}\,d_{v,(j,k)}\,x_{i,(j,k)}\tag{$O_1$} \label{O1}\]
\begin{align}
    \sum_{i =1}^p x_{i,(j,k)} &= 1 & \forall (j,k) \in E  \tag{$S_1$} \label{S1} \\
    \sum_{v \in V} w_{iv} &= 1  & \forall i\in \{1,2,...p\}  \tag{$S_2$}\label{S2} \\
    \sum_{(j,k) \in E} b_{(j,k)} x_{i,(j,k)} &\leq  \Phi_{u} & \forall i \in \{1,2,...p\}  \tag{$B_1$}\label{B1} \\
    \sum_{(j,k) \in E} b_{(j,k)} x_{i,(j,k)} &\geq \Phi_{l}  & \forall i \in \{1,2,...p\}  \tag{$B_2$}\label{B2} \\
    \sum_{(u,v) \in \delta(V_D)} x_{i,(u,v)}+ \sum_{e \in D} (1 - x_{i,e})+\sum_{v\in V_D} w_{iv}   &\geq 1 & \forall D \subseteq E,\ \forall i \in \{1,\dots,p\}\tag{$C_1$}\label{c1} \\
    x_{i,(j,k)} &\in \mathbb{B} & \forall i\leq p, \forall (j,k) \in E  \tag{$I_1$} \label{A2}\\
    w_{ik} &\in \mathbb{B} & \forall i \leq p, \forall k \in V  \tag{$N_1$}\label{N1} 
\end{align}

More detailed explanations of the constraint groups and associated model components are given in the corresponding paragraphs ($S$, $B$, $C$, $I$, $N$, $O$, $W$). 

To improve readability in the remainder of this paper, we introduce the shorthand notation $c_i \in V$ to denote the center node representing district $i$. Formally, this means that $w_{i c_i} = 1$ and $w_{ik} = 0$ for all $k \neq c_i$. While the selection of centers is fully encoded by the binary variables $w_{ik}$ in BP-EBD, using $c_i$ allows for more concise expressions in proofs and discussions. Unless explicitly stated otherwise, references to $c_i$ should be understood as referring to the unique node $v \in V$ such that $w_{iv} = 1$.

\subsection{Groupings of Districting Design Requirements}

This paper investigates \emph{why} and \emph{when} the districting problem becomes computationally hard. While districting is known to be \NP\ in general (see ~\cite{kassem2023models}), we study what happens when specific criteria are removed (i.e., constraints are relaxed). Our goal is to precisely identify the boundary between tractable instances (in \( \mathcal{P} \)) and intractable ones (\NP). To this end, we define a family of \emph{subproblems} by omitting selected criteria groups, as previously categorized. The full model is denoted by \( \theta_{S,B,C,I,N,O,W} \), indicating that all design criteria groups are simultaneously imposed. Since the structural requirements (\( S \)) are essential for a well-defined model, we assume they are always included and henceforth omit them from our notation for brevity; thus, we refer to the full model as \( \theta_{B,C,I,N,O,W} \) and consider any meaningful subproblems resulting from omitting constraint groups.

In the following, we discuss the implications of relaxing the full model and solving the resulting subproblem. In some cases, constraints can be dropped without replacement; in others, a suitable alternative must be introduced (e.g., if removing the objective, a proxy criterion may be required to preserve meaning). We begin with the structural constraints, showing that omitting them either trivializes the problem or leads to formulations that are not meaningful or well-posed. Moreover, we discuss why some constraints only make sense when combined with other groups and, conversely, why certain  combinations are not meaningful.

\subsection*{(S) Structural}

Constraint \ref{S1} enforces that every edge must be allocated a district. This is a fundamental requirement, as omitting it renders the problem meaningless: setting all \( x_{i,(j,k)} = 0 \) yields a trivial solution with an objective value of zero, which is not informative or useful.

Constraint \ref{S2} enforces that each district selects exactly one center node as its representative. Together with other constraints, it ensures that exactly \( p \) center nodes are opened. Without this restriction, the model could theoretically assign each edge to its own district with a unique center (e.g., one of its incident nodes), leading to \( |E| \) centers. This would trivially satisfy balance and contiguity, and minimize the objective to zero—again, defeating the purpose of meaningful districting.

These structural constraints are essential for defining the problem correctly and are therefore always included. Hence, we omit the symbol \( S \) when listing constraint groups with the featured notation, thereby assuming their inclusion by default.

\subsection*{(B) Balance}

Constraints \ref{B1} and \ref{B2} enforce balance of each district. Specifically, \ref{B1} ensures that no district exceeds the allowed upper bound, while \ref{B2} guarantees a minimum summed weight for each district. If balance is not a requirement, both constraints can be omitted without the need of replacements. The resulting model remains feasible but may produce districts that differ significantly in size and weight.

Let the average district weight be denoted as 
\[
\overline{b} = \frac{\sum_{(j,k) \in E} b_{(j,k)}}{p}.
\]
It is clear that if the lower bound \(\Phi_\ell > \overline{b}\) or the upper bound \(\Phi_u < \overline{b}\), then no feasible solution exists. Hence, we always assume \(\Phi_\ell \leq \overline{b}\leq \Phi_u\).

A common approach is to model balance using a tolerance parameter \(\tau\), which specifies how much a district may deviate from the average. We feature two ways of imposing balance:

\begin{itemize}
    \item[(a)] \textbf{Additive Tolerance:} Replace the right-hand side of \ref{B1} with \( (1+\tau)\overline{b} \), and the right-hand side of \ref{B2} with \( (1-\tau)\overline{b} \),  where $\tau\geq 0$.
    \item[(b)] \textbf{Multiplicative Tolerance:} Replace the right-hand side of \ref{B1} with \( \frac{\overline{b}}{\tau} \), and the right-hand side of \ref{B2} with \( \tau\cdot \overline{b} \), where $\tau\in (0,1]$.
\end{itemize}

Under variant (a), perfect balance is achieved with \(\tau = 0\), and under variant (b), it is achieved with \(\tau = 1\). The additive form is more widely used in districting literature \citep{rios2020research}, as it is easier to interpret and tune. However, it may come with practical drawbacks. Namely, for regions with large variation in district sizes, the setting \(\tau > 1\) may be needed to ensure feasibility, which implies tolerating districts more than twice the average size and possibly a negative lower bound.

The multiplicative form, while less common, has the advantage of always producing non-negative bounds. However, since \(\tau\) appears in the denominator for the upper bound, it cannot be set to zero, and thus each district must have positive total weight.  In other words, we must ensure that every district contains at least one edge (which becomes relevant in Theorem~\ref{lem:B}).

Because each balance variant comes with trade-offs, we opt for a general formulation using parameters \(\Phi_\ell\) and \(\Phi_u\), which can flexibly model either case. These parameters can be derived from any \(\tau\), additive or multiplicative. In our \NP{}ness proofs, we only rely on the ability to enforce that each district contains exactly \(\overline{b}\) edges which is a property supported by both balance requirement variants.

\subsection*{(C) Contiguity}

This constraint is only meaningful when allocations are integral. Otherwise, it is possible to allocate any edge to any district at $\epsilon$ value, fulfilling the contiguity constraint without increasing the objective.

To enforce contiguity in edge-based districting, it is necessary to ensure that for any group of edges \( D \subseteq E \) forming an isolated component, if all edges in \( D \) are assigned to district \( i \), and the center node \( c_i \notin V_D \) (i.e., the center is not part of that component), then at least one edge in district \( i \) must connect \( D \) to the remainder of the graph.

We use constraint \ref{c1}, which is given by:

\begin{equation*}
 \sum_{(u,v) \in \delta(V_D)} x_{i,(u,v)}+ \sum_{(u,v) \in D} (1 - x_{i,(u,v)})+\sum_{v\in V_D} w_{iv}   \geq 1 \quad \forall D \subseteq E,\ \forall i \in \{1,\dots,p\}
\end{equation*}

Notice that each of the three terms is non-negative. Hence, the formula excludes the case where all sums are equal to zero, which occurs exclusively when three conditions are simultaneously met: \begin{itemize}
    \item \textit{No} edge connects \( D \) to the remainder of the graph (otherwise, $x_{i,(u,v)}=1$ for some edge $(u,v)$).
    \item \textit{All} edges in \( D \) are allocated to district $i$ (otherwise, $(1 - x_{i,(u,v)})=1$ for some edge $(u,v)$).
    \item \textit{No} vertex incident to one edge in $D$ is the center node of $i$ (otherwise some $w_{iv}=1$).  
\end{itemize}

In practice, one would not enumerate all subsets \( D \subseteq E \), as there is an exponential number of possibilities. Instead, such cut constraints are typically added lazily during a branch-and-cut procedure. An alternative is to model contiguity via flow-based connectivity models, but these tend to perform poorly computationally \citep{val22imp} and are thus not widely utilized in practice.

\subsection*{(I) Integrity}

Constraint \ref{A2} ensures that each edge is assigned exclusively to a single district. An alternative relaxed version would allow fractional assignment of edges to multiple districts. In that case, the binary constraint is replaced as:

\[
x_{i,(j,k)} \geq 0 \quad \forall i \in V,\; \forall (j,k) \in E. \tag{$I_1'$}
\]

Note that, due to Constraint \ref{S1}, the upper bound \( x_{i,(j,k)} \leq 1 \) is already implicitly enforced and need not be stated explicitly.

\subsection*{(N) Node Center Selection}

Our model allows multiple districts to have the same node as their center. This differs from the original formulation by Kassem \& Escobedo~\cite{kassem2023models} and is justified for two key reasons. First, enforcing uniqueness of center nodes can artificially increase problem difficulty. In particular, if an arbitrary number of centers can be opened while not considering balance or continuity, the problem should be easy; however this is not true if the centers are required to be unique. Moreover, enabling one node to serve as the center of multiple districts generalizes the model and avoids pathological cases. Second, in practical applications such as locating polling places, postboxes, or facilities, multiple services can operate from the same location without conflict. However, it is admittedly rare that two (or more) optimal centers share a node in practical applications. In all, this modification contributes more to theoretical bounds, but it can at times yield unrealistic solutions.

We also consider the setting without node center selection, such that the set of center nodes $\{c_1,c_2,...,c_p\}= C$ are given in advance, in which case the variables \( w_{ik} \) become parameters. The constraint \ref{N1}, which is related to center selection, can be omitted in this setting. 

\subsection*{(O) Compactness}

The objective function \ref{O1} minimizes the total distances of assigned edges to their respective center nodes, promoting compact district shapes. In certain problem variants, one may be interested only in the feasibility of a solution—whether a valid districting exists—regardless of compactness. In this case, the objective can be dropped and replaced with a dummy objective such as:
\[
\min\; 0.  \tag{$O_1'$}
\]

This effectively transforms the problem into a pure feasibility check, where the goal is to determine whether a solution exists that satisfies all remaining constraints.

This change of objective is only meaningful in the scenario of contiguity. In all other cases, greedily assign the edges in arbitrary fashion to centers and may check for balance reasons if the number of edges is divisible by $p$. Hence, we only consider dropping this criterion when contiguity is imposed.

\subsection*{(W) Weightedness}

In the weighted setting, each edge is assigned a positive real weight \( b \in \mathbb{R}^+ \). We refer to this setting as the \emph{weighted case}, denoted by \( \theta_W \). Additionally, we consider the \emph{unweighted case}, in which each edge \((j,k)\) is assigned a unit weight, i.e., \( b_{(j,k)} = 1 \) for all edges.

\subsection{Extending the Results to Node-Based Districting}

As discussed in Section~\ref{sec:Re}, we can consider node-based variants of the featured class of districting problems. In this case, nodes (rather than edges) are assigned to districts, but connectivity is still enforced through an edge-based approach. This hybrid model, which we refer to as \emph{edge-based connectivity for assignments}, introduces two sets of binary decision variables:

\begin{itemize}
    \item \( x^\textbf{V}_{i,j} \in \{0,1\} \): Indicates whether node \( j \in V \) is allocated to district \( i \in \{1, \dots, p\} \).
    \item \( x^\textbf{E}_{i,(j,k)} \in \{0,1\} \): Indicates whether edge \( (j,k) \in E \) is allocated to district \( i \) (to enforce connectivity).
\end{itemize}

We emphasize that the superscripts \textbf{V} and \textbf{E} in the above notation are used to distinguish those variables that deal with allocating vertices and those that deal with allocating edges. The cost function is updated accordingly, specifically by excluding costs associated with edges. Let \( d_{i,j} \in \mathbb{R}_{\geq 0} \) represent the distance between node \( j \) and district \( i \). This replaces the edge-based cost terms in the original formulation.

To adapt BP-EBD to the node-based setting:
\begin{itemize}
    \item In all constraints and in the objective that involve assignments (e.g., \eqref{O1}, \eqref{S1}, \eqref{B1}, \eqref{B2}, \eqref{A2}), replace the edge assignment variable \( x_{i,(j,k)} \) with \( x^\textbf{V}_{i,j} \), and replace the corresponding cost coefficients \( d_{i,(j,k)} \) with the node-based values \( d_{i,j} \).
    \item Connectivity constraints are modeled separately using the edge variables \( x^e \), as described next.
\end{itemize}

\paragraph{Connectivity Enforcement}  
Because we consider edge-based connectivity for node assignments, it is necessary to allocate edges to districts to ensure contiguity. Hence we have two kind of $x$ variables and use the edge-based variables \( x^\textbf{E} \) to enforce that the nodes assigned to each district form a connected component.

The key idea is to require that, if all nodes in a subset \( S \subseteq V \) are assigned to the same district \( i \), and the center node \( c_i \notin S \), then some edge assigned to district \( i \) must connect \( S \) to the remainder of the graph. As in the connectivity section, we define:

\begin{itemize}
    \item \( S \subset V \): a subset of nodes that does not contain the center node \( c_i \).
    \item \( \delta(S) := \{ (u,v) \in E \mid u \in S,\ v \notin S \} \): the set of edges crossing from \( S \) to \( V \setminus S \).
\end{itemize}

We can then replace \ref{c1} with the constraint,
\begin{equation}
\sum_{(u,v) \in \delta(S)} x^\textbf{E}_{i,(u,v)} + \sum_{v \in S} (1 - x^\textbf{V}_{i,v})+\sum_{v\in S} w_{iv} \geq 1
\quad \forall S \subset V\  \forall i \in \{1,\dots,p\}.
\label{con:node-cut}
\end{equation}

This inequality becomes active only if all nodes in \( S \) are assigned to district \( i \) (i.e., \( x^\textbf{V}_{i,v} = 1 \) for all \( v \in S \)) and $c_i \not \in S$ (otherwise $\sum_{v\in S} w_{iv}>0$). In that case, the left-hand side forces at least one edge on the cut between \( S \) and \( V \setminus S \) to be included in district \( i \), thus ensuring that \( S \) is connected to the rest of the district (specifically to the center \( c_i \)).

As with the edge-based model, the number of such constraints is exponential in \( |V| \) and, thus, are typically handled using lazy constraint generation in large-scale instances or replaced with alternative flow-based formulations. Our theoretical results are developed under the edge-based formulation. However, the key arguments can be adapted to this node-based setting.

\section{Subproblem Complexity Classification}\label{sec:Result}

The main contribution of this paper is summarized in Table~\ref{table}, which classifies each variant of the districting problem by its computational complexity—either polynomial-time solvable or \NP. Interestingly, with the exception of contiguity, each individual constraint can tip the complexity of a problem from polynomial-time solvable to \NP. Table \ref{table2} highlights, for each constraint, an example of a problem that is polynomial solvable, but becomes \NP\ upon the addition of that particular constraint. Another key finding is that the combination of node center selection and compactness leads to an \NP\ problem.

\begin{table}[ht]
\centering
\begin{tabular}{ll}
\toprule
Variant solvable in polynomial time& \NP\ superset \\
\midrule
$\theta_{I,O,W}$   & $\theta_{\textcolor{red}{B},I,O,W}$ \\
$\theta_{B,I,O}$   & $\theta_{B,\textcolor{red}{C},I,O}$  \\
$\theta_{B,O,W}$   & $\theta_{B,\textcolor{red}{I},O,W}$ \\
$\theta_{O}$       & $\theta_{\textcolor{red}{N},O}$      \\
$\theta_{C,I,N,W}$ & $\theta_{C,I,N,\textcolor{red}{O},W}$\\
$\theta_{B,I,O}$   & $\theta_{B,I,O,\textcolor{red}{W}}$  \\
\bottomrule
\end{tabular}
\caption{Any additional constraint can make a polynomial time solvable subproblem \NP.}\label{table2}

\end{table}

Proofs for the entries in Table~\ref{table} are either given directly or need additionally a reduction between two districting variants. For instance, Theorem~\ref{lem:AWCO} establishes that $\theta_{I,O,W}\in \mathcal{P}$. By reduction~\ref{lem:W}, removing the weight constraint does not make the problem harder, so $\theta_{I,O}\in \mathcal{P}$ as well. We refer to these proofs in the corresponding  $\theta_{I,O}$ part in the table as (Theorems) [\ref{lem:AWCO} \& \ref{lem:W}].

%Afterwards we show \NP ness for specific problems. In the chapter afterwords we show the reductions. Moreover, we show that for $\Phi_u\geq\frac{2|E|}{3}, \Phi_l\leq\frac{|E|}{3}$, $p=2$ and $|E|\geq2$, $\theta_{B,C,I,N}$ is always efficient feasible solvable. This is considered in \ref{lem:Bonus} and has an extra section dedicated to it.
\begin{table}[h]
\centering
\caption{All 32 meaningful combinations of design criteria \(B,C,I,N,O,W\); P refers to polynomial-time, and H refers to \NP.}
\label{table}
\renewcommand{\arraystretch}{1.2}
\small
\begin{tabular}{c|cccc}
\toprule
\textbf{\(B, N, W \setminus C, I, O\)} & O & IO & IC & CIO \\
\midrule
--   & P [\ref{WBCO}] & P [\ref{lem:AWCO} \& \ref{lem:W}] & P [\ref{lem:AWCO} \& \ref{lem:W} \& \ref{lem:O}] & P [\ref{lem:AWCO} \& \ref{lem:W}] \\
W    & P [\ref{WBCO}] & P [\ref{lem:AWCO}] & P [\ref{lem:AWCO} \& \ref{lem:O}] & P [\ref{lem:AWCO}]\\
B    & P [\ref{WBCO}]& P [\ref{lem:ABO}] & H [\ref{lem:WABC}] & H [\ref{lem:WABC} \& \ref{lem:O}] \\
BW   & P [\ref{WBCO}] & H [\ref{lem:WABC}] & H [\ref{lem:WABC} \& \ref{lem:W}] & H [\ref{lem:WABC} \& \ref{lem:O}] \\
N    & H [\ref{lem:ANO}] & H [\ref{lem:ANO}] & P [\ref{lem:WNAC} \& \ref{lem:W}] & H [\ref{lem:ANO}]\\
NW   & H [\ref{lem:ANO} \& \ref{lem:W}] & H [\ref{lem:ANO} \& \ref{lem:W}] & P [\ref{lem:WNAC}] & H [\ref{lem:ANO} \& \ref{lem:W}]  \\
BN   & H [\ref{lem:ANO} \& \ref{lem:B}] & H [\ref{lem:ANO} \& \ref{lem:B}] & H [\ref{lem:WABC}] & H [\ref{lem:WABC} \& \ref{lem:O}] \\
BNW  & H [\ref{lem:ANO} \& \ref{lem:W} \& \ref{lem:B}] & H [\ref{lem:ANO} \& \ref{lem:W} \& \ref{lem:B}] & H [\ref{lem:WABC} \& \ref{lem:W}] & H [\ref{lem:WABC} \& \ref{lem:O} \&\ref{lem:W}] \\
\bottomrule
\end{tabular}
\end{table}

The rest of this section is organized as follows. Section \ref{sec:P} demonstrates polynomial time solvability of certain subproblems. Then, Section \ref{sec:np} demonstrates \NP ness for others. Finally, Section \ref{sec:Reduc} covers the complexity of subproblems that can be determined based on the established complexity results of the first two subsections. 

\subsection{Subproblems with Polynomial Time Complexity}\label{sec:P}

\begin{theo}
Let $\mathcal{H}=\{B,O,W\}$ be a set of design criteria, denoting its  power set as $\mathcal{P}(\mathcal{H})$. For any $X\in \mathcal{P}(\mathcal{H})$, the induced problem  $\theta_{X}$ is solvable in polynomial time.\label{WBCO}
\end{theo}\label{LPTheorem}

Proof: We linearize the components of BP-EBD and give an equivalent formulation. To do so, it is necessary to designate an additional center node index with the allocation decisions. The resulting variable, denoted as $y$, is given by:
\[y_{v,i,(j,k)} = \begin{cases} 1 \text{ if edge }(j,k) \text{ is assigned to center node } v \text{ and district $i$;}\\ 
\text{0 otherwise.}
\end{cases}\]
The resulting formulation is given by: 

\[ \min \sum_{v\in V}\sum_{i=1}^p\sum_{(j,k)\in E} d_{v,(j,k)}\,y_{v,i,(j,k)} \]
\begin{align*}
    \sum_{v\in V}\sum_{i=1}^p y_{v,i,(j,k)} &= 1 & \forall (j,k) \in E \\
    \sum_{v \in V} w_{iv} &= 1  & \forall i\in \{1,2,...p\}   \\
    y_{v,i,(j,k)} &\leq w_{v,i} & \forall (j,k) \in E, i\in \{1,2,...p\}, v\in V   \\
    \sum_{(j,k) \in E}\sum_{v\in V} b_{(j,k)} y_{v,i,(j,k)} &\leq  \Phi_{u} & \forall i \in \{1,2,...p\}   \\
    \sum_{(j,k) \in E}\sum_{v\in V} b_{(j,k)} y_{v,i,(j,k)} &\geq \Phi_{l}  & \forall i \in \{1,2,...p\}   \\
    \sum_{(u,v) \in \delta(V_D)}\sum_{v\in V} y_{v,i,(u,v)}+ \sum_{e \in D} \left(1 - \sum_{v\in V} y_{v,i,e}\right)+\sum_{v\in V_D} w_{iv}   &\geq 1 & \forall D \subseteq E,\  i \in \{1,\dots,p\} \\
    y_{v,i,(j,k)} &\in [0,1] & \forall i\leq p, (j,k) \in E  \\
    w_{ik} &\in [0,1] & \forall i \leq p, k \in V, v\in V  
\end{align*}

The allocation variables, encapsulated as the vector $y$, serve to remove the previous multiplication of decision variables in objective \eqref{O1}. The constraint $y_{v,i,(j,k)} \leq w_{v,i}$ is added to  impose the requisite logic, namely, that an allocation to center node $v$ district $i$ can only be done when the respective location variable is activated (previously enforced through the multiplication of the associated binary variables). All other constraints simply reflect the change in allocation variables. Finally, note that the domain of variables $y$ and $w$ is relaxed, since the Integrity (I) and Node Center Selection (N) requirements are excluded from $\mathcal{H}$. This makes the above   formulation a linear program and, therefore, in \Pclass~. \qed

\begin{theo}
    The variants $\theta_{C,I,O,W}$ and $\theta_{I,O,W}$ are in \Pclass. \label{lem:AWCO}
\end{theo}

Since the center nodes are fixed in $\theta_{I,O,W}$, it is optimal to greedily assign each edge (or vertex, in the case of node-based districting) to its closest center node. We define an arbitrary but consistent tie-breaking rule based on center indices: for edges equidistant to multiple centers, assign them to the center with the smallest index. We now claim that this assignment also satisfies the contiguity requirement, meaning that each district forms a connected subgraph and therefore solves $\theta_{C,I,O,W}$.

Assume, for contradiction, that the resulting assignment does \emph{not} yield connected districts. Then, there exists an edge $e$ assigned to district $k$, such that on \emph{any} direct path from $e$ to its assigned center $c_k$, there exists an intermediate edge $e'$ which is assigned to a different district, say $k'$. By \eqref{firsteq}, we can write the distance between $c_k\in C$ and $e$ as $d_{c_k,e}=d_{c_k,e'}+b_{e'}+d_{e',e}$.

Since $e'$ is assigned to $k'$, we must have that $d_{c_{k'},e'}< d_{c_k,e'} $ or $d_{c_{k'},e'}= d_{c_k,e'} $; with ${k'} < k$ satisfies the tie-breaking rule. Next, the distance from $e$ to $c_{k'}\in C$ is calculated as:
\[ d_{c_{k'},e} \]
\[ \leq  d_{c_{k'},e'}+b_{e'}+d_{e',e} \]
\[ = d_{c_k,e}-d_{c_k,e'}+ d_{c_{k'},e'}\]
\[ \leq d_{c_k,e}, \]
which holds at equality only if ${k'} < k$. Thus, either $d_{c_{k'},e}<d_{c_k,e}$ or $d_{c_{k'},e}=d_{c_k,e}$, with ${k'}<k$. Both cases contradict the assumption that $e$ is assigned to $c_k$.

Hence, the greedy assignment ensures that each edge is connected to its center via edges in the same district, and the resulting solution is contiguous. Since all steps are polynomial in $|E|$, the problem $\theta_{C,I,O,W}$ can be solved in polynomial time.
\qed

\begin{theo}
    The variant $\theta_{B,I,O}$ is in \Pclass. \label{lem:ABO}
\end{theo}

Proof: As edges are allocated fully to centers, we can round $\Phi_l$ and $\Phi_u$ to their nearest lower and higher integer values, respectively, via the functions 
$\left\lceil \Phi_{l} \right\rceil$ and $\left\lfloor \Phi_{u} \right\rfloor$. Next, relax the integrity constraint (I) by allowing the allocation variables to take fractional values, that is, set the domain of variable $x$ as 
\[
0 \leq x_{i,(j,k)} \leq 1 \quad \forall i \leq p,\;  (j,k) \in E,
\]
in place of the original binary restriction \(x_{i,(j,k)} \in \{0,1\}\). 

To obtain a binary solution, we iteratively round the fractional LP solution without increasing the objective value. At each iteration, we eliminate at least one fractional variable \(x_{i,(j,k)} \in (0,1)\) by constructing a bipartite \emph{rounding graph} $G= (U \cup W, F)$ where:
    \begin{itemize}
        \item \(U = E\) is the set of \emph{edge-vertices}, each representing an edge \((j,k) \in E\) (resp. nodes for node-based districting),
        \item \(W = C\) is the set of \emph{center-vertices}, each representing a center \(c_i \in C\),
        \item There exists an edge \(((j,k), i) \in F\) if and only if \( x_{i,(j,k)} \in (0,1) \), i.e., the allocation  of edge \((j,k)\) to center \(i\) is currently fractional.
    \item Each edge in \(F\) is weighted with the corresponding value \(x_{i,(j,k)}\). 
    %\item Edges with \(x_{i,(j,k)} = 0\) or \(x_{i,(j,k)} = 1\) are omitted from the graph, as these are fixed and no longer subject to change during the rounding process.
\end{itemize}

We now distinguish three structural cases for the rounding graph \(G\). In each case, apply a local weight adjustment to the values of the allocation variables, \(x_{i,(j,k)}\), to reduce the number of fractional variables while maintaining feasibility and the same objective function value. We showcase this procedure using the example rounding graph in Figure \ref{fig:enter-label}.

\begin{figure}[ht]
    \centering
    
    % First subfigure
    \begin{subfigure}{0.3\textwidth}
        \centering
        \begin{tikzpicture}[scale=1, transform shape]
            \tikzstyle{every node}=[draw, circle, minimum size=8mm]
            
            % Nodes
            \foreach \x in {1,2,3} 
                \node (u\x) at (0, -\x) {$e_\x$};
            \foreach \y in {1,2,3} 
                \node (v\y) at (3, -\y) {$c_\y$};
            
            % Edges
            \draw (u1) -- (v1);
            \draw (u1) -- (v2);
            \draw (u2) -- (v1);
            \draw (u2) -- (v2);
            \draw (u3) -- (v3);
            \draw (u3) -- (v2);
        \end{tikzpicture}
        \caption{Example Binary Graph.}
        \label{fig:enter-label}
    \end{subfigure}
    \hfill
     \begin{subfigure}{0.3\textwidth}
        \centering
        \begin{tikzpicture}[scale=1, transform shape]
            \tikzstyle{every node}=[draw, circle, minimum size=8mm]
            
            % Nodes
            \foreach \x in {1,2,3} 
                \node (u\x) at (0, -\x) {$e_\x$};
            \foreach \y in {1,2,3} 
                \node (v\y) at (3, -\y) {$c_\y$};
            
            % Edges
            \draw[red, dashed] (u1) -- (v1);
            \draw[blue, dash dot] (u1) -- (v2);
            \draw[blue, dash dot] (u2) -- (v1);
            \draw[red, dashed] (u2) -- (v2);
            \draw (u3) -- (v3);
            \draw (u3) -- (v2);
        \end{tikzpicture}
        \caption{Red-Blue Cycle.}
        \label{fig:2}
    \end{subfigure}
\hfill
    % Second subfigure
    \begin{subfigure}{0.3\textwidth}
        \centering
        \begin{tikzpicture}[scale=1, transform shape]
            \tikzstyle{every node}=[draw, circle, minimum size=8mm]
            
            % Nodes
            \foreach \x in {1,2,3} 
                \node (u\x) at (0, -\x) {$e_\x$};
            \foreach \y in {1,2,3} 
                \node (v\y) at (3, -\y) {$c_\y$};
            
            % Edges
            \draw[red, dashed] (u1) -- (v1);
            \draw[blue, dash dot] (u1) -- (v2);
            \draw[blue, dash dot] (u3) -- (v3);
            \draw[red, dashed] (u3) -- (v2);
        \end{tikzpicture}
        \caption{Red-Blue path.}
        \label{fig:3}
    \end{subfigure}
    
    % Third subfigure
   
    \caption{Example of a rounding graph and identified substructures for eliminating fractional variables}
    \label{fig:all}
\end{figure}
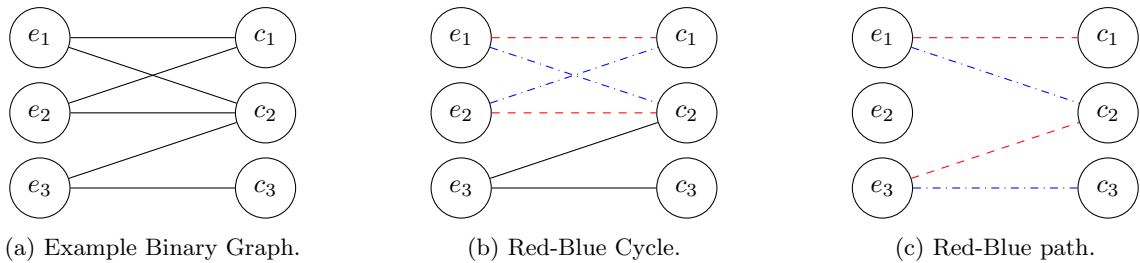

\textit{Case 1: The rounding graph contains a cycle.}

Since the graph is bipartite, all cycles are of even length. Without loss of generality, select one of the cycles and color its edges alternating between red and blue (see  Figure~\ref{fig:2}, for a small example). Next, increase all red edge weights by \(\varepsilon\), and decrease all blue ones by \(\varepsilon\), such that the value of each remains in $[0,1]$. Importantly, since any node is connected to the same number of blue and red edges, this adjustment neither changes $\sum_{i=1}^px_{i,(j,k)}$ (\ref{S1}) nor $\sum_{v\in V}w_{v,i}$ (\ref{S2}), thereby guaranteeing that the new values of $x$ constitute a valid solution. The objective changes by:

\[
\varepsilon \left( \sum_{((j,k),i) \in F_{\text{red}}} d_{c_i,(j,k)} - \sum_{((j,k),i) \in F_{\text{blue}}} d_{c_i,(j,k)} \right).
\]

If this term is less than zero, this indicates that the solution can be improved, contradicting optimality. If it is bigger than zero, the solution can also be improved by  interchanging the colors blue and red. Therefore, it must be exactly zero. Note that choosing the largest possible \(\varepsilon\) forces at least one variable to 0 or 1, reducing the number of fractional variables, as planned.

\textit{Case 2: The graph is a forest with at least one edge.}

By property \ref{S1}, no edge-vertex can have degree 1. By the definition of a forest, there exists a path between two center-vertices of degree 1. In similar fashion as Case 1, select one and color its edges,  alternating between red and blue (see Figure~\ref{fig:3}), and applying the same \(\pm \varepsilon\) adjustment.

Unlike Case 1, here the total assignment for the two endpoint centers changes slightly, but only by adjusting a single variable. Since the balance bounds \(\Phi_l, \Phi_u\) are integral and the only a single  fractional-valued edge of a center-vertex is adjusted, the bounds cannot be violated. This again reduces the number of binary variables.

\textit{Case 3: No edge remains.}

In this situation, all \( x_{c,e} \in \{0,1\} \), and the solution is integral and feasible.

Let \(\mathcal{F} \subseteq C \times E\) denote the initial set of allocation variables with fractional values. Since there are at most \(|C| \cdot |E|\) such variables, the number of iterations required to fully round the solution is at most \(|\mathcal{F}| \leq |C| \cdot |E|\), which is polynomial in the input size. Hence, the rounding process terminates in polynomial time, and  \(\theta_{B,I,O}\) is in \Pclass.
\qed

To extend the result to node-based districting, linearize $x^\textbf{E}$ instead and follow an analogous line of reasoning as the preceding proof. The resulting linear program can be solved in polynomial time.

\begin{theo}
    The variant $\theta_{C,I,N,W}$ is in \Pclass. \label{lem:WNAC}
\end{theo}

Proof: Compactness (i.e., the objective function) and balance are not enforced in this problem. Hence, it is possible to allocate every edge (resp. node) to a single district, say district 1, by setting \[\quad  x_{1,(j,k)}=1 \quad\ \forall (j,k)\in E\; \text{, and }\] \[x_{i,(j,k)}=0  \quad \forall (j,k)\in E\quad \forall i>1.\] The center nodes can be chosen arbitrarily, and the resulting solution is valid. \qed

In Lemma~\ref{lem:Bonus}, we extend this variant by incorporating balance constraints and investigate how tightly these constraints can be imposed while still guaranteeing the existence of a valid solution.

\subsection{Subproblems with {\NP } Complexity}\label{sec:np}

\begin{theo}
    The variants $\theta_{N,O}$, $\theta_{I,N,O}$, and $\theta_{C,I,N,O}$ are \NP.\label{lem:ANO}
\end{theo}

Proof: We show that the classical \textsc{vertex cover} problem —one of Karp’s well-known 21 $\mathsf{NP}$-complete problems \cite{karp2009reducibility}— can be reduced from an instance of the edge-based districting problem $\theta_{I,N,O}$ or $\theta_{C,I,N,O}$. In the vertex cover problem, the goal is to select a subset of vertices $U \subseteq V$ of minimum cardinality such that every edge $e\in E$ has at least one endpoint in $U$.

Suppose we have an algorithm that can solve $\theta_{I,N,O}$ optimally in polynomial time. Then, we can use it to solve vertex cover by evaluating $\theta_{I,N,O}$, for $p = 1, 2, \ldots, |V|$. The smallest value of $p$ for which the objective value is $\sum_{(j,k)\in E}d_{j,(j,k)}$ corresponds to an optimal vertex cover of size $p$, as this is a solution where any edge is connected to an adjacent center node. This argument also holds for $\theta_{C,I,N,O}$ as a vertex cover solution is also connected.   

Hence, solving $\theta_{I,N,O}$ or $\theta_{C,I,N,O}$ efficiently would yield a polynomial-time algorithm for vertex cover, implying that $\theta_{I,N,O}$ is \NP.

We now argue that dropping Integrity ($I$) does not change the objective value and therefore, does not make the problem easier to solve. To this end, we allow each edge to be fractionally assigned to multiple centers. This change cannot yield an invalid solution since neither balance nor contiguity are being enforced. 

For any optimal solution where an edge $(j,k)$ is assigned fractionally to two districts $i_1$ and $i_2$, it must hold that $d_{i_1,(j,k)}=d_{i_2,(j,k)}$ or, otherwise, it would be suboptimal to assign part of the edge to the more distant of the two. Nonetheless, this also means that the edge could instead be allocated in full to one of these centers, which enables a fractional value to be converted into an integral one without changing the objective value. Therefore, the relaxed version is equivalent in complexity to the original, and it follows that  $\theta_{N,O}$ remains \NP.\qed

To extend the result to node-based districting, perform an analogous reduction using the \NP \textsc{ dominating set} problem instead of vertex cover. In dominating set, a graph $G=(V,E)$ is given and the problem is to give a minimal subset $D\subseteq V$ such that for every vertex $u\in V \backslash D$ there is a vertex $v\in D$ such that $(u,v)\in E$ (see \cite{garey1979computers}). A dominating set problem solution of size $p$ directly refers to a $\theta_{I,N,O}$ solution with the same $p$ and objective value of $|V|-p$.

\begin{theo}
   The variants $\theta_{B,C,I}$, $\theta_{B,C,I,N}$ and $\theta_{B,I,O,W}$ are \NP . \label{lem:WABC}
\end{theo}

We give a polynomial-time reduction to the \textsc{3-Partition} problem, which is known to be strong-$\mathsf{NP}$-complete (see \cite{garey1979computers}). Given a set of positive integers $\mathcal{S} = \{s_1, s_2, \dots, s_n\}$, where $n=3m$, the \textsc{3-Partition} problem seeks to partition $\mathcal{S}$ into $m$ disjoint subsets $S_1, \dots, S_m$ such that the sum of the elements in each subset is equal. Therefore, we must have that
\[
\sum_{s \in S_1} s = \sum_{s \in S_2} s = \dots = \sum_{s \in S_m} s =\sum_{s \in \mathcal{S}} s\cdot \frac{1}{m}=T.
\]

For any integer $s$, the relationship $\frac{T}{4}<s<\frac{T}{2}$ holds and, thus, it directly follows that each set in the disjunction contains exactly $3$ integers.

We now construct an instance of $\theta_{B,C,I}$ from a $3$-Partitioning problem using Algorithm \ref{alg:1}. 
\begin{algorithm}[H]
\caption{Graph Construction to reduce from $\theta_{B,C,I}$ to 3-Partitioning.} \label{alg:1}
\begin{algorithmic}[1]
\Require Set of positive integers $\mathcal{S} = \{s_1, s_2, \dots, s_n\}$
\State Initialize vertex set $V = \{v_0\}$
\State Initialize edge set $E = \emptyset$
\State $p \gets m$
\State $\Phi_u \gets \Phi_l \gets T$
\For{$i = 1$ to $n$}
    \State $V \gets V \cup v_{i,1}$ 
    \State $E \gets E \cup (v_{i,1}, v_0)$ 
    \For{$j = 2$ to $s_i$}
        \State $V \gets V \cup v_{i,j}$ 
        \State $E \gets E \cup (v_{i,j}, v_{i,j-1})$  
    \EndFor
\EndFor
\State \Return Graph $G = (V, E)$ with weights and center assignments

\end{algorithmic}
\end{algorithm}

Any solution to the constructed instance of $\theta_{B,C,I}$ must assign $T$ edges to every district. Notice that if $(v_{i,1},v_0)$ is assigned to district $j$ ($x_{j,(v_{i,1},v_0)}=1$), it also holds that \[x_{j,(v_{i,k},v_{i,k-1})}=1 \quad \quad \forall\quad  1<k\leq S_i, \]

meaning that any edge along one “arm" is assigned to the same district. This is due to the connectivity requirement. Otherwise, a subset of these edges need to form an independent district, which is not possible as there are less than $\frac{T}{2}$ edges, per the above note. Thus, any feasible solution to this instance corresponds directly to a solution to the 3-Partition problem and vice versa. Figure \ref{figure25} visualizes the transformation for a 3-Partition instance.

\begin{figure}[H]
    \centering
    \begin{tikzpicture}
        % center node
        \node[shape=circle,draw=black] (112) at (0,5) {$v_0$};
    
        % leaves
        \node[shape=circle,draw=black] (1) at (-5,3) {};
        \node[shape=circle,draw=black] (11) at (-5,2) {};
        \node[shape=circle,draw=black] (12) at (-5,1) {};
        \node[shape=circle,draw=black] (13) at (-5,0) {};
        \node[shape=circle,draw=none] (16) at (-5,-2.5) {$4$};

        \node[shape=circle,draw=black] (2) at (-3,3) {};
        \node[shape=circle,draw=black] (21) at (-3,2) {};
        \node[shape=circle,draw=black] (22) at (-3,1) {};
        \node[shape=circle,draw=black] (23) at (-3,0) {};
        \node[shape=circle,draw=black] (24) at (-3,-1) {};
        \node[shape=circle,draw=none] (26) at (-3,-2.5) {$5$};

        \node[shape=circle,draw=black] (3) at (-1,3) {};
        \node[shape=circle,draw=black] (31) at (-1,2) {};
        \node[shape=circle,draw=black] (32) at (-1,1) {};
        \node[shape=circle,draw=black] (33) at (-1,0) {};
        \node[shape=circle,draw=black] (34) at (-1,-1) {};
        \node[shape=circle,draw=none] (36) at (-1,-2.5) {$5$};
        
        \node[shape=circle,draw=black] (4) at (1,3) {};
        \node[shape=circle,draw=black] (41) at (1,2) {};
        \node[shape=circle,draw=black] (42) at (1,1) {};
        \node[shape=circle,draw=black] (43) at (1,0) {};
        \node[shape=circle,draw=black] (44) at (1,-1) {};
        \node[shape=circle,draw=none] (46) at (1,-2.5) {$5$};
        
        \node[shape=circle,draw=black] (5) at (3,3) {};
        \node[shape=circle,draw=black] (51) at (3,2) {};
        \node[shape=circle,draw=black] (52) at (3,1) {};
        \node[shape=circle,draw=black] (53) at (3,0) {};
        \node[shape=circle,draw=black] (54) at (3,-1) {};
        \node[shape=circle,draw=none] (56) at (3,-2.5) {$5$};

        \node[shape=circle,draw=black] (6) at (5,3) {};
        \node[shape=circle,draw=black] (61) at (5,2) {};
        \node[shape=circle,draw=black] (62) at (5,1) {};
        \node[shape=circle,draw=black] (63) at (5,0) {};
        \node[shape=circle,draw=black] (64) at (5,-1) {};
        \node[shape=circle,draw=black] (65) at (5,-2) {};
        \node[shape=circle,draw=none] (66) at (5,-2.5) {$6$};

        % edges
        \draw[dashed, purple] (112) -- (1);
        \draw[-, green!50!black] (112) -- (2);
        \draw[-, green!50!black] (112) -- (3);
        \draw[dashed, purple] (112) -- (6);
        \draw[dashed, purple] (112) -- (5);
        \draw[-, green!50!black] (112) -- (4);

        \draw[-, green!50!black] (2) -- (21);
        \draw[-, green!50!black] (3) -- (31);
        \draw[-, green!50!black] (4) -- (41);
        \draw[dashed, purple] (1) -- (11);
        \draw[dashed, purple] (6) -- (61);
        \draw[dashed, purple] (5) -- (51);
        
        \draw[dashed, purple] (11) -- (12);            
        \draw[dashed, purple] (12) -- (13);

        \draw[dashed, purple] (61) -- (62);            
        \draw[dashed, purple] (62) -- (63);
        \draw[dashed, purple] (63) -- (64);            
        \draw[dashed, purple] (52) -- (53);
        \draw[dashed, purple] (51) -- (52);            
        \draw[dashed, purple] (54) -- (53);
     
        \draw[-, green!50!black] (42) -- (43);
        \draw[-, green!50!black] (42) -- (41);
        \draw[-, green!50!black] (44) -- (43);
        \draw[dashed, purple] (65) -- (64);

        \draw[-, green!50!black] (22) -- (23);
        \draw[-, green!50!black] (22) -- (21);
        \draw[-, green!50!black] (24) -- (23);

        \draw[-, green!50!black] (32) -- (33);
        \draw[-, green!50!black] (32) -- (31);
        \draw[-, green!50!black] (34) -- (33);
    \end{tikzpicture}
    \caption{A transformation from the 3-Partition problem of $S=\{4,5,5,5,5,6\}$ to $\theta_{B,C,I}$; the colors show a solution which directly correspond to the disjunction $S_1=\{4,5,6\}$ and $S_2=\{5,5,5\}$.}  
    \label{figure25}
\end{figure}
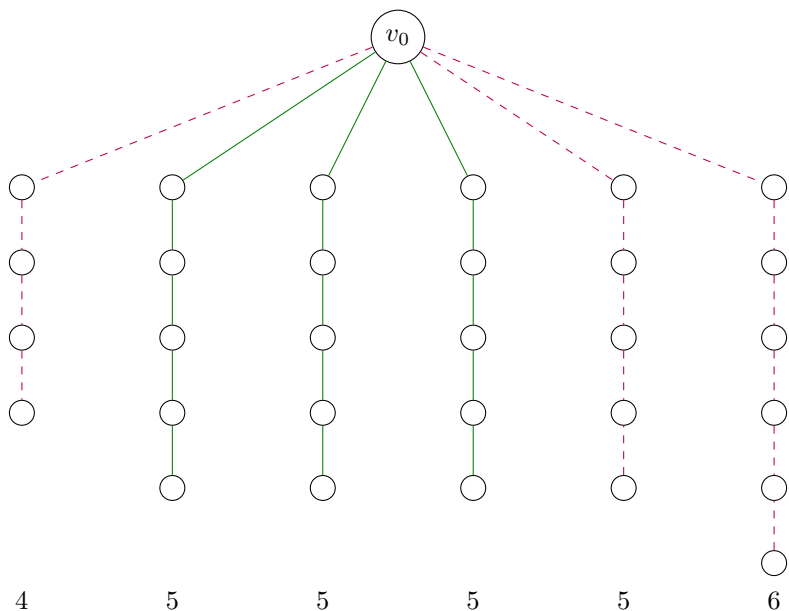

Since 3-Partitioning is strongly-\NP, the problem remains \NP\ if any $s_i$ is bounded polynomially in $n$; additionally, \NP ness is preserved through the reduction by using $T\cdot M$ edges (see \cite{andreev2004balanced}). Therefore, solving $\theta_{B,C,I}$ is at least as hard as solving 3-Partition. 

Determining the center node locations is easy in this setting, as they can all be placed at $v_0$ (recall that our model assumes that multiple districts can shared the same node as their center). Hence this discussion also holds for $\theta_{B,C,N,I}$.

This proof can be further adapted to $\theta_{B,I,W}$, by setting $b_{(v_0,v_{i,1})}=s_i$ and including only those vertices that are directly incident to $v_0$. In other words, vertices $v_{i,j}$ with $j>1$ are excluded form this construction (see Figure \ref{figure252}, for a small example). The previous reasoning holds and, thus $\theta_{B,I,W}$ is also \NP. Finally, because the introduction of the objective function cannot make the problem easier, the problem $\theta_{B,I,O,W}$ is also \NP. \qed

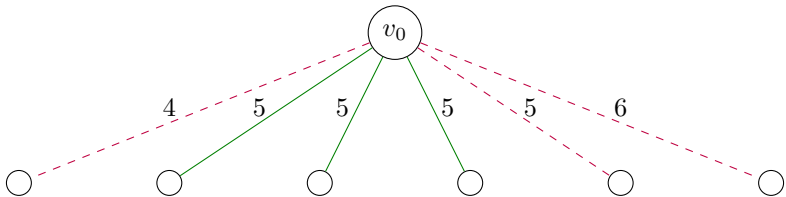
\begin{figure}[H]
    \centering
    \begin{tikzpicture}
        % center node
        \node[shape=circle,draw=black] (112) at (0,5) {$v_0$};
    
        % leaves
        \node[shape=circle,draw=black] (1) at (-5,3) {};
        \node[shape=circle,draw=none] (16) at (-3,4) {$4$};

        \node[shape=circle,draw=black] (2) at (-3,3) {};
        \node[shape=circle,draw=none] (26) at (-1.8,4) {$5$};

        \node[shape=circle,draw=black] (3) at (-1,3) {};
        \node[shape=circle,draw=none] (36) at (-0.7,4) {$5$};
            
        \node[shape=circle,draw=black] (4) at (1,3) {};
        \node[shape=circle,draw=none] (46) at (0.7,4) {$5$};
            
        \node[shape=circle,draw=black] (5) at (3,3) {};
        \node[shape=circle,draw=none] (56) at (1.8,4) {$5$};

        \node[shape=circle,draw=black] (6) at (5,3) {};
        \node[shape=circle,draw=none] (66) at (3,4) {$6$};

        % edges (purple dashed, green solid)
        \draw[dashed, purple] (112) -- (1);
        \draw[-, green!50!black] (112) -- (2);
        \draw[-, green!50!black] (112) -- (3);
        \draw[dashed, purple] (112) -- (6);
        \draw[dashed, purple] (112) -- (5);
        \draw[-, green!50!black] (112) -- (4);
    \end{tikzpicture}
    \caption{A transformation from the 3-Partition problem of $S=\{4,5,5,5,5,6\}$ to $\theta_{B,I,W}$. 
    The colors indicate the partition: dashed purple edges correspond to one subset ($S_1=\{4,5,6\}$), solid green edges to the other subset ($S_2=\{5,5,5\}$).}  
    \label{figure252}
\end{figure}

\section{Reductions between districting problems}\label{sec:Reduc}

This section demonstrates that certain districting criteria can only increase the difficulty of the problem. Together with proofs of sections \ref{sec:P} and \ref{sec:np}, we can apply the findings  to show \NP ness or polynomial time solvability for a broader range of problems. Let $\mathcal{F}=\{B,C,I,N,O,W\}$ be the six districting criteria considered in this work, denoting its power set as $\mathcal{P}(\mathcal{F})$. As before, let $X\in \mathcal{P}(\mathcal{F})$ represent any subset on design criteria and $\theta_X$ be the associated problem.

\begin{theo}
    The complexity relationship $\theta_{X,N}\leq_P \theta_{X,B,N}$ holds.\label{lem:B}
\end{theo}

Proof: To show this, we need to introduce Balance constraints, which does not make the problem easier. For the additive version of Balance (B), set $\Phi_l=0$ and $\Phi_u=\sum_{(j,k) \in E} b_{(j,k)}$ and we are done. For the multiplicative version, setting $\Phi_l=0$ is disallowed as setting $\tau=0$ gives a division by $0$. Instead, for this case, set $\Phi_l=1$ through $\tau=\frac{1}{|E|}$. We show that for Node Center Selection (N), there is an optimal solution containing one edge per district. Otherwise, we could iteratively choose one district with at least two edges, choose one edge such that the remaining part stays connected (if necessary), and then choose an adjacent node as center. 

In node-based districting, this is even more straightforward. In this setting, we can promote vertices to become centers, since  connecting a vertex to itself no added cost.\qed 

Interestingly, applying the  same reasoning for the multiplicative version of Balance is less clear for the complexity relationship; meaning $\theta_{X}\leq_{P}^? \theta_{X,B}$ is unclear. By construction, we have that $\overline{b} \cdot \tau>0$. However, we can think of cases where any optimal solution for $\theta_{B,I,O}$ does not assign an edge to each center. The problem evolves when multiple centers are at the same vertex as visualized in Figure \ref{figure36}. Here two centers are located in the right vertex, and it is optimal to assign the left two edges to the left center instead.

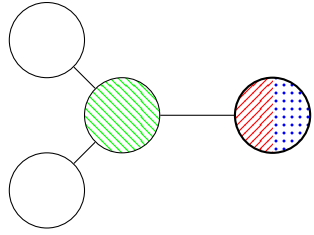
\begin{figure}[H]
    \centering
    \begin{tikzpicture}

    \def\radius{0.5cm}
    \node[shape=circle,fill=white] (2324) at (-0.32,0) {};
    % Rote Hälfte (links) komplett gefüllt + Muster
    \begin{scope}
        \clip (0,0) circle(\radius);
        \fill[red, pattern=north east lines, pattern color=red] (-\radius,\radius) rectangle (0, -\radius);
    \end{scope}

    % Blaue Hälfte (rechts) komplett gefüllt + Muster
    \begin{scope}
        \clip (0,0) circle(\radius);
        \fill[blue, pattern=dots, pattern color=blue] (0,\radius) rectangle (\radius, -\radius);
    \end{scope}

    % Kreisrand
    \draw[thick] (0,0) circle(\radius);

    % Grüner Knoten mit Muster + schwarzem Rand
    \node[shape=circle,minimum size=1cm, fill=green!30, pattern=north west lines, pattern color=green, draw=black] (112) at (-2,0) {};

    % Schwarze Knoten
    \node[shape=circle,minimum size=1cm,draw=black, fill=white] (2) at (-3,-1) {};
    \node[shape=circle,minimum size=1cm,draw=black, fill=white] (1) at (-3,1) {};

    % Kanten
    \draw[-] (112) -- (1);
    \draw[-] (112) -- (2);
    \draw[-] (2324) -- (112);

    \end{tikzpicture}
    \caption{Any optimal solution would not assign an edge to either the red or the blue center.}  
    \label{figure36}
\end{figure}

\begin{theo}
    The complexity relationship $\theta_X\leq_P \theta_{X,O}$ holds.\label{lem:O}
\end{theo}
Proof: A compact (i.e., distance-optimal) solution contains the information that a feasible solution exists and is, therefore, the former is least as hard to find. \qed
\begin{theo}
    The complexity relationship $\theta_X\leq_P \theta_{X,W}$ holds.\label{lem:W}
\end{theo}
Proof: The unweighted problem is a special case of the weighted problem (i.e., obtained by setting all  weights to 1).\qed

\section{Parameterizing the Complexity of the Balance Criterion}\label{sec:Bonus}

Up to this point, we have examined specific problems and classified them as either polynomially tractable or \NP. Naturally, once the objective function is dropped, the problem becomes substantially easier, as demonstrated in Theorem~\ref{lem:WNAC}, where a valid solution can always be constructed. Furthermore, Theorem~\ref{lem:B} shows that balance constraints can be chosen in a non-restrictive way. In particular, it follows that $\theta_{B,C,I,N}$ with additive tolerance $\tau=0$ always admits a valid solution for any connected graph $G=(V,E)$. This motivates the question of how restrictive $\tau$ (or, equivalently, $\Phi$) can be set while still guaranteeing the existence of a feasible solution. We address this question in the case of two districts.

\begin{lem}
Consider the problem $\theta_{B,C,I,N}$, with $p = 2$ and $|E| \geq 2$ on a weighted and connected graph $G=(V,E)$. A feasible solution can always be found efficiently for $\Phi_u\geq\frac{2|E|}{3}$ and $\Phi_l\leq\frac{|E|}{3}$. When one bound is exceeded, at least one unsolvable instance exists.
\label{lem:Bonus}
\end{lem}

\begin{proof}
For $|E|\leq 3$ this claim holds trivially. Hence, this proof  focuses on instances with $|E|> 3$.

We first show how to construct a feasible solution. To this end, construct a tree-based auxiliary structure, denoted as the graph $G_T=(V_T, E_T)$, as follows. Begin by setting $V_T$ initially to the original vertex set $V$ and set $E_T=\emptyset$. Next, iterate over all edges in $E$ and insert each into $E_T$, as long as the additional does not form a cycle. If inserting an edge $(u,v)$ would create a cycle in $G_T$, introduce a new auxiliary vertex $v'$ and connect it to one endpoint of $(u,v)$ (chosen arbitrarily), and add $v'$ to $V_T$. This ensures that $G_T$ remains a tree and that each edge is included exactly once. By construction, any adjacent edges in $G_T$ are also adjacent in the original graph, and any edge-based allocations on this tree corresponds to a valid allocation in $G$. Notice that $|E|=|E_T|$.

The purpose of $G_T$ is to reduce any possible graph structure to trees. As that tree could have been the input graph in the first place, this does not restrict the results. In the resulting tree, there exists a unique vertex incident to edges assigned to both $c_1$ and $c_2$. We refer to this vertex as the \emph{cut-point}.

Next, to facilitate balanced districting, evaluate vertices $v\in V_T$ as potential cut-points and analyze the subtrees that would result from their removal. For a fixed cut point, each incident edge can be associated with a different connected subtree (component) once that vertex is removed. The number of edges in each such subtree is recorded with a label next to each edge that was originally incident to the cut-point. Figure~\ref{figure6} illustrates this idea via a small numerical example (the black vertex represents the cut-point).

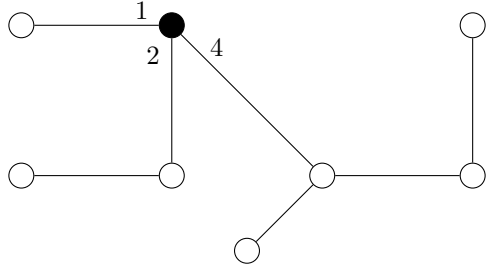
\begin{figure}[H]
	\centering
	\begin{tikzpicture}
        \node[shape=circle,draw=black, fill=black] (112) at (0,5) {};
        \node[shape=circle,draw=black] (2) at (-2,3) {};
        \node[shape=circle,draw=black] (1) at (0,3) {};
        \node[shape=circle,draw=black] (3) at (2,3) {};
        \node[shape=circle,draw=black] (4) at (4,3) {};
        \node[shape=circle,draw=black] (5) at (4,5) {};
        \node[shape=circle,draw=black] (6) at (1,2) {};
        \node[shape=circle,draw=black] (7) at (-2,5) {};
        \node[shape=circle] (label1) at (-0.4,5.2) {1};
        \node[shape=circle] (label2) at (0.6,4.7) {4};
        \node[shape=circle] (label3) at (-0.25,4.6) {2};
        \draw[-] (1) -- (2);
        \draw[-] (112) -- (1);
        \draw[-] (112) -- (3);
        \draw[-] (5) -- (4);
        \draw[-] (3) -- (4);
        \draw[-] (3) -- (6);
        \draw[-] (112) -- (7);
	\end{tikzpicture}
	\caption{The edge labels represent the size of the connected component after the cut-point (the black vertex) is removed.}
    \label{figure6}
\end{figure}

We iteratively distinguish three cases:

1. \textbf{One component exceeds $\frac{2|E|}{3}$ edges:}  
   Let edge $e$ be the edge of the largest component that  is incident to the cut-point. We relocate the cut point to be the other vertex incident to $e$. This reduces its size of the largest component by at least one. Any smaller components size is still bounded by $\frac{|E|}{3}+1<\frac{2|E|}{3}$ edges. We repeat this process until no component exceeds $\frac{2|E|}{3}$ edges.

2. \textbf{One component contains between $\frac{|E|}{3}$ and $\frac{2|E|}{3}$ edges:}  
   Assign this component to district $c_1$, and assign the remaining edges to $c_2$.

3. \textbf{All components are smaller than $\frac{|E|}{3}$:}  
   Iteratively assign components to $c_1$ until their combined size reaches at least $\frac{|E|}{3}$. As each individual component is smaller than $\frac{|E|}{3}$, their total cannot exceed $\frac{2|E|}{3}$. Assign the remaining edges to $c_2$.

After executing the procedure in  case 2 or 3, each district has between $\frac{|E|}{3}$ and $\frac{2|E|}{3}$ edges and is valid.

To show the bound is tight, consider Figure \ref{figure5} where any allocation forces one district to take 2 arms.

\begin{figure}[H]
		\centering
		\begin{tikzpicture}

            \node[shape=circle,draw=black] (112) at (0,5) {};
            \node[shape=circle,draw=black] (2) at (-2,3) {};
            \node[shape=circle,draw=black] (1) at (0,3) {};
            \node[shape=circle,draw=black] (3) at (2,3) {};
		      \node[shape=circle] (21) at (-2,2.2) {.};
               \node[shape=circle] (21) at (-2,2.4) {.};
                \node[shape=circle] (21) at (-2,2) {.};
            \node[shape=circle] (11) at (0,2) {.};
            \node[shape=circle] (11) at (0,2.2) {.};
            \node[shape=circle] (11) at (0,2.4) {.};
            \node[shape=circle] (31) at (2,2) {.};
            \node[shape=circle] (31) at (2,2.2) {.};
            \node[shape=circle] (31) at (2,2.4) {.};
			
			\draw[-,red, dashed] (112) -- (2);
            \draw[-,red, dashed] (112) -- (1);
            \draw[-,blue] (112) -- (3);

		\end{tikzpicture}
		\caption{A graph with three arms, each with the same number of edges; for this graph, the problem  $\theta_{B,C,I,N}$ has no feasible solutions when $\Phi_u<\frac{2|E|}{3}$ or $\Phi_l>\frac{|E|}{3}$.}  
        \label{figure5}
	\end{figure}
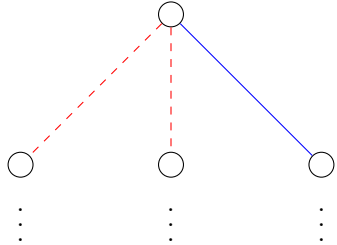

%This result can be used iteratively for $p=2^x$ districts.

For the additive version of Balance, the claim is applicable to $\tau=\frac{1}{3}$ as \[\left(1+\frac{1}{3}\right)\cdot \overline{b}=\frac{2|E|}{3}\quad and\quad \left(1-\frac{1}{3}\right)\cdot \overline{b}=\frac{|E|}{3}. \]

For the multiplicative version of Balance, the claim is applicable to $\tau=\frac{2}{3}$ as \[\frac{\overline{b}}{\frac{2}{3}}>\frac{2|E|}{3}\quad and \quad \frac{2}{3}\cdot \overline{b}=\frac{|E|}{3}. \]
\end{proof}

The example with the multi-armed graph can be generalize to $p$ districts by $p+1$ arms. One district must take two arms, giving $\Phi_l\geq \frac{2|E|}{p+1}$, while the average district load is $\overline{b} = \frac{|E|}{p}$. For the additive version, this generalizes to

\[
\frac{2|E|}{p+1} \leq \frac{\overline{b}}{\tau} 
\quad \Rightarrow \quad
\tau \leq \lim_{p \to \infty} \frac{p}{p+1} \cdot \frac{1}{2} = \frac{1}{2}.
\]

\section{Conclusion}\label{sec:con}
This paper addresses a fundamental question: under which conditions do districting problems become computationally intractable? To answer it, we systematically analyze the computational complexity tipping points of logistics districting by grouping commonly used criteria/requirements  and examining the complexity of their combinations. We classify each resulting subproblem as either solvable in polynomial time or \NP, via formal mathematical proofs. This comprehensive classification offers practical insights into how one might design districting systems by strategically relaxing certain constraints to ensure tractability without fully compromising the usefulness of the solution. 

Beyond its theoretical contributions, this work has important real-world implications. The precise sources of computational hardness enable a deeper understanding of important trade-offs. Practitioners can now better identify which requirements are essential and which may be safely omitted to achieve efficient and scalable solutions.

While our study focuses on the dichotomy between polynomial-time solvability and \NP ness, many \NP \ variants may still admit high-quality approximate solutions. Investigating approximation algorithms and performance guarantees for these hard cases is a promising direction for future research. Such work would further bridge the gap between theoretical intractability and practical usability, helping to develop robust, near-optimal tools for complex districting tasks.

\bibliographystyle{agsm}
 \bibliography{bib}

\end{document}